\documentclass{article}
\usepackage{fullpage}
\usepackage[affil-it]{authblk}

\usepackage{amsmath,amsthm,amssymb,bm,mathtools}
\usepackage{hyperref}
\usepackage{cite,wasysym}
\usepackage{microtype}
\usepackage{color}

\newcommand{\red}{\color{red}}

\makeatletter
\def\thanks#1{\protected@xdef\@thanks{\@thanks
        \protect\footnotetext{#1}}}
\makeatother

\newtheorem{theorem}{Theorem}
\newtheorem{prop}{Proposition}

\theoremstyle{definition}

\usepackage{bm}
\usepackage{float}

\newcommand{\normal}{\mathcal{N}}
\newcommand{\adj}{A}

\newcommand{\reals}{\mathbb{R}}

\DeclareMathOperator{\diag}{\mathrm diag}
\newcommand{\one}{\bm{1}}

\newcommand{\ex}[1]{\ensuremath{\mathbb{E}\left[ #1\right]}}

\title{Gaussian Mixture Models for  Stochastic Block Models\\ with Non-Vanishing Noise}

\author{Heather Mathews$^{\star}$  
}
\author{Vaishakhi Mayya$^{\dagger}$}
\author{Alexander Volfovsky$^{\star}$}
\author{Galen Reeves$^{\star\dagger}$\thanks{This work was partially supported by funding from the Laboratory for Analytic Sciences (LAS), the Army Research Institute (ARI) under grant number W911NF1810233 and the NSF under Grant No. 1750362.} }
\affil{$^{\star}$Statistical Science and $^{\dagger}$Electrical and Computer Engineering, Duke University}

\begin{document}

\maketitle
\begin{abstract}
Community detection tasks have received a lot of attention across statistics, machine learning, and information theory with a large body of work concentrating on theoretical guarantees for the stochastic block model. 
One line of recent work has focused on modeling the spectral embedding of a network using Gaussian mixture models (GMMs) in scaling regimes where the ability to detect community memberships improves with the size of the network. However, these regimes are not very realistic. This paper provides tractable methodology motivated by new theoretical results for networks with non-vanishing noise. We present a procedure for community detection using GMMs that incorporates certain truncation and shrinkage effects that arise in the non-vanishing noise regime. We provide empirical validation of this new representation using both simulated and real-world data. 
 \end{abstract}



\section{Introduction}

Network data are of paramount importance across many modern scientific fields \cite{farine2015constructing,dunbar2015structure,stadtfeld2019integration,grinberg2019fake}. One of the most common tasks in network analysis is the search for community structure among units in the network. Much of the statistical \cite{Rohe11, athreya2016limit,Suwan16} and information theoretical  \cite{decelle2011asymptotic,abbe:2018,lelarge:2018,reeves2019} work on community detection 
studies the stochastic block model (SBM) \cite{Holland83}. In this probabilistic network model, the probability of an edge between individuals $i$ and $j$ is governed exclusively by their community memberships, $X_i$, and $X_j$. The complete network can then be represented by its adjacency matrix, $\adj$, where $A_{ij}=A_{ji}=1$ if there is an edge between node $i$ and node $j$ and $0$ otherwise. In this setting, the task of {\it community detection} is to recover the community labels $X = (X_1, \dots, X_n)$ given the adjacency matrix, $\adj$, and possible side information.

\begin{figure}[!h]
\begin{center}
\includegraphics[width= .6\columnwidth,clip=true, trim=0 0 0 0]{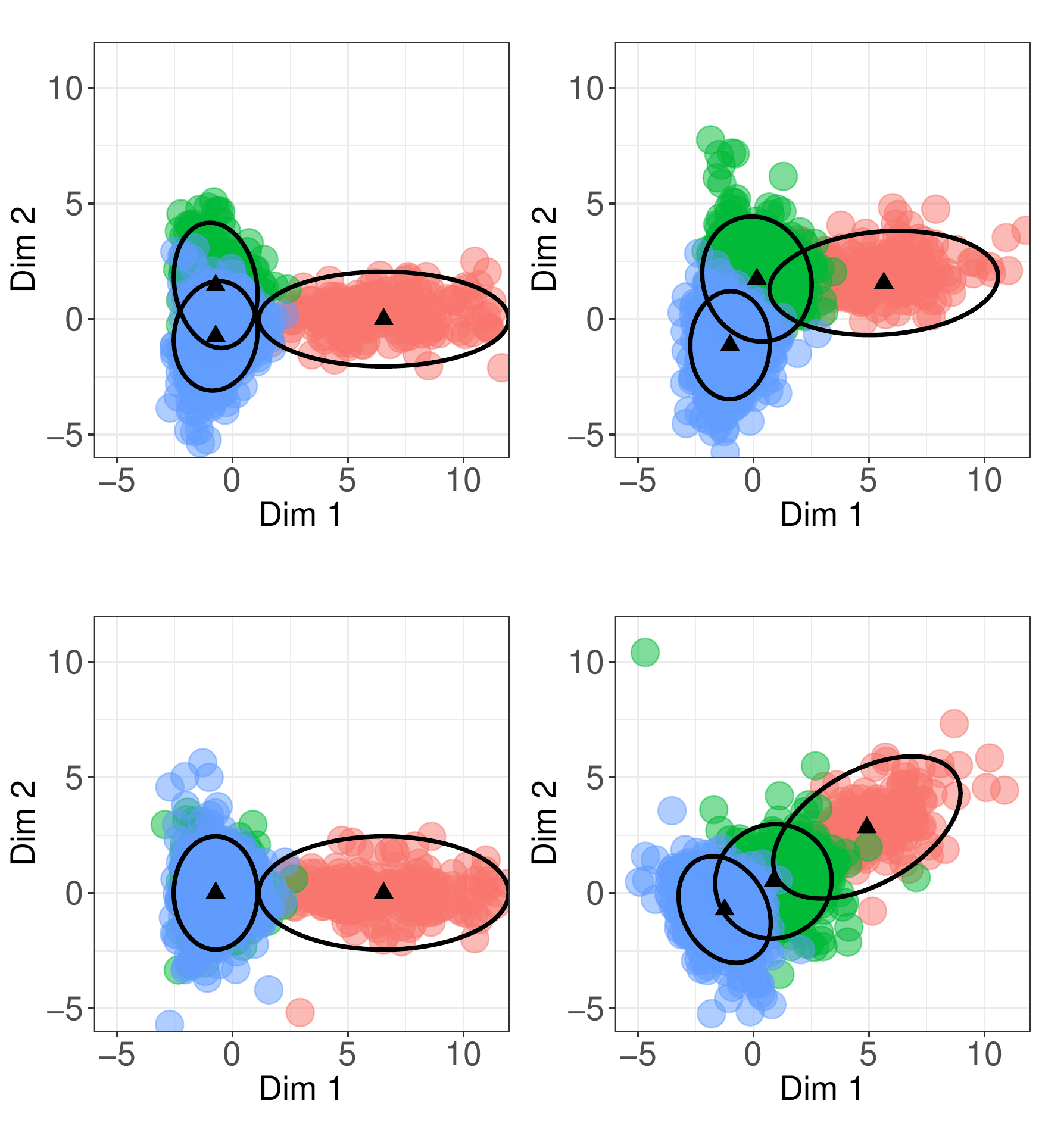}
\caption{Aligned eigenvectors from graphs generated from four different SBMs. The models in the top row have two eigenvalues of $R$ that are greater than one while the bottom row only has one eigenvalue greater than one. In the left column $R$ is diagonal, while in the right column there is a rotation by an orthogonal matrix. Each $\blacktriangle$ and ellipse are based on the mean and covariance of GMM components from Section~\ref{sec:methodmain}.}
\label{fig:rotating_R_gmm_covar}
\end{center}
\end{figure}

A large body of work has considered spectral clustering methods for community detection with early work focusing on the behavior of clustering nodes via $k$-means \cite{Rohe11} and more recent work focusing on Gaussian mixture models (GMMs)~\cite{athreya2016limit,Suwan16}. Much of this work is derived under scaling regimes where the ability to detect community memberships improves with the size of the graph. {\it We call this the vanishing-noise regime}.

In this paper,  we focus on the degree-balanced SBM in which each community has the same expected average degree $d$. In Section~\ref{sec:dbsbm}, we provide a formal problem formulation that describes how all the relevant information can be captured in a $(K-1)$ dimensional embedding of the eigenvectors of the adjacency matrix. We then develop a novel GMM representation for mid- to high-noise regimes that is able to appropriately quantify the uncertainty about the labels of the individual nodes. In Figure~\ref{fig:rotating_R_gmm_covar}, we display our proposed eigenvector embedding and the mean and covariance matrices we derive in Section~\ref{sec:methodmain} for our GMM across models with both high- (lower row of plots) and low-noise (upper row of plots) levels. Section~\ref{sec:sims} provides an empirical validation of our method and a comparison to other state-of-the-art algorithms.

\section{Problem Formulation}\label{sec:dbsbm}

Our approach is described in the context of a general latent space model for networks that includes the degree balanced SBM as a special case. Let $X = (X_1, \dots, X_n)$ be an $n \times s$ matrix of latent variables whose rows are drawn independently from a distribution $P$ on $\reals^s$ with mean zero, identity  covariance, and bounded support.  Conditional on   $X$, the entries of the adjacency matrix $A = (A_{ij})$ are drawn independently according to
\begin{align}
    A_{ij} \sim \text{Bern}\left( \frac{d}{n}  + \frac{\sigma}{n} X_i^T R X_j \right), \quad i < j, \label{eq:A}
\end{align}
where: $d$ is a positive number that parameterizes the expected degree of each node in the network, $R$ is a symmetric $s \times s$ matrix that describes how the probability of an edge depends on the latent variables, and $ \sigma = \sqrt{d (n-d)/n}$  is a scaling factor that ensures that the signal-to-noise ratio is invariant to the choice of $d$. The tuple $(n,P,d,R)$ is valid only if $d/n  + (\sigma/n) x^T R \tilde{x}$ is between zero and one for all $x,\tilde{x}$ in the support of $P$. 

The assumption that $P$ has zero mean  ensures that the model is degree-balanced in expectation. Specifically, $\ex{A_{ij} \mid X_i} = d/n$, and thus the expected degree of node $i$ is independent of $X_i$. The assumption that $P$ has identity covariance is without loss of generality since any linear transformation of the latent variables $X_i$ can be absorbed into the model parameter $R$. 

When the support of $P$ is finite, this model reduces to the degreee balanced SBM where each node is assigned to exactly one of $K$ possible communities, independently with probability vector $p  = (p_1, \dots, p_K)$. Specifically, we let $X_i$ be supported on a set of  $K$ points $\{\mu_1, \dots, \mu_K\}$ in dimension $s = K-1$ satisfying the moments constraints:
\begin{align}
\sum_{k=1}^K p_k \mu_k = 0, \qquad \sum_{k=1}^K p_k \mu_k \mu_{k}^T  = I. \label{eq:mu_white}
\end{align}
We leverage a unique specification of $P$ given in \cite[Remark 1]{reeves2019} as a function of label probabilities $p$.  

We note that an alternative representation for the labels, used in previous work~\cite{Rohe11}, is to associate the $k$-th label with the $k$-th standard basis vector in $\reals^K$. An explicit mapping between these various representations is provided in \cite{reeves2019}.

\section{Theory and Methodolgy}\label{sec:methodmain}

In this section, we present a general method for inference in degree-balanced networks. This method is based on a Gaussian approximation for the spectral embedding of the adjacency matrix. It is well known that the leading eigenvector of the adjacency matrix is correlated with the degree of the nodes and thus does not provide any information about the memberships. Therefore, we consider the spectral embedding of the normalized adjacency matrix $A - (d/n)\one \one^T$.  The projection of this matrix onto the space of rank-$s$ matrices can be expressed as $V\Lambda V^T$  where $\Lambda = \diag(\lambda_1, \dots, \lambda_s)$ contains the largest  eigenvalues (in magnitude), in decreasing order,  and  $V = (V_1, \dots, V_n)^T$ is an $n \times s$ matrix with orthonormal columns corresponding to the eigenvectors. We note that the representation of the eigenvectors is not unique. 

\subsection{Gaussian Approximation of Eigenvectors}

In the context of community detection, the basic principle underlying spectral clustering is that the rows of the leading eigenvectors are correlated with the latent variables. Applying standard clustering techniques, such as $k$-means, directly on the points $V_1, \dots, V_n \in \reals^s$ provides a partition of the nodes in the network and can be used to estimate community memberships. {\it More generally, a principled approach to inference is to formulate a joint model for the eigenvectors and the latent parameters}. This applies in the general latent space model as well as in the specialized case of community detection where the parameter space is finite. This paper builds upon recent work \cite{Suwan16,athreya2016limit,reeves2019}, which provides both theoretical and empirical support for the use of GMMs.

To describe our approach,  we introduce the scaled eigenvectors $Y_1, \dots, Y_n \in \reals^s$ according to
\begin{align}
Y_i =  \sqrt{n} \diag(r_1, \dots, r_s)V_i, \label{eq:scaled_vecs}
\end{align}
where $r_1 \ge r_2 \ge \dots r_s$ are the eigenvaues of $R$. The results in Athreya et al.~\cite[Theorem~4.8]{athreya2016limit} can be used to characterize the asymptotic distribution of $Y_i$ in the vanishing-noise regime where $P$ is fixed while the model parameters $(n,d,R)$ scale to infinity.  Adapted to the setting of this paper, this result suggests the following approximation:
\begin{align}
    U Y_i \sim \normal\left( R X_i, \tilde{\Sigma}(X_i) \right), \label{eq:Gaussian_approx_athreya}
\end{align}
where $U$ is an orthogonal matrix that aligns the eigenvectors with the latent variables and the covariance is given by
\begin{align}
\tilde{\Sigma}(x) &  =\mathbb{E}_{X_0\sim P}\left[ \nu(x, X_0)   X_0 X_0^T\right]\\
\nu(x,\tilde{x})& = 1 +    \left( \frac{n  - 2 d}{n \sigma } \right) x^T R \tilde{x}   - \frac{1}{n} \left( x^T R \tilde{x} \right)^2.
\end{align}
The matrix $U$ depends on the eigenspace of $R$ as well as the particular choice of eigenvectors used in the eigendecomposition of $A$.  
In the proposed method described below, this matrix is estimated from the data.

It is important to emphasize that the approximation in \eqref{eq:Gaussian_approx_athreya} is adapted from the vanishing-noise regime where the eigenvalues of $R$ scale with $n$. As a consequence, some important aspects of the  moderate to high noise regimes are not captured. In particular, it is well known that an eigenvector is uninformative about latent structure unless its associated eigenvalue exceeds a threshold.

We propose a Gaussian approximation for the scaled eigenvectors $Y_i$ that {\it incorporates both truncation and shrinkage effects via the mean and variance of $Y_i$}. Let $\bar{R}$ and $\underline{R}$ be the symmetric $s \times s$ matrices obtained by applying the mappings $r \mapsto \max(|r|,1)$  and $r \mapsto \min(|r|,1)$, respectively, to the eigenvalues of $R$. Our approximation is given by
\begin{align}
    U Y_i \sim \normal\left( (\bar{R}^2 - I)^{1/2} X_i, \Sigma(X_i) \right), \label{eq:Gaussian_approx}
\end{align}
where $U$ is an orthogonal matrix and the covariance is
\begin{align}
    \Sigma(x) & = (I \! -\! \bar{R}^{-2})^{-1/2} \tilde{\Sigma}(x) (I\!  -\! \bar{R}^{-2})^{-1/2}  + \bar{R}^{-1} \underline{R}^2 \bar{R}^{-1}. \label{eq:Sigma_x}
\end{align}
The term $(\bar{R}^2 - I)^{1/2}$ provides the truncation and shrinkage to the eigenvalues of $R$. 
Note that any direction corresponding to an eigenvalue of magnitude less than one does not provide any information about $X_i$.

Our approximation follows from  a leave-one-out argument combined with asymptotic properties of spiked Wigner matrices \cite{Georges2012}. 

\begin{prop}\label{prop:Sigma_x}
The matrices $\Sigma(x)$ and $\tilde{\Sigma}(x)$ satisfy
\begin{align}
   \Sigma(x) &= \tilde{\Sigma}(x)  + O(\lambda_\mathrm{min}^{-2}(R)) \label{eq:Sig_bigR}\\
      \Sigma(x)& = \underline{R}^2   + O(\sigma^{-1}   + n^{-1} ) \label{eq:Sig_bigd}\\
    \tilde{\Sigma}(x)& = \underline{R}^2   + O(\sigma^{-1}   + n^{-1} ).
\end{align}
\end{prop}
\begin{proof}[Proof Sketch]
These results follow straightforwardly from the assumption that $P$ has bounded support and the fact that $\nu(x,\tilde{x}) = 1  + O(\sigma^{-1} + n^{-1})$. 
\end{proof}

Proposition~\ref{prop:Sigma_x} shows that as $R$ increases, the approximation in \eqref{eq:Gaussian_approx} converges to the vanishing-noise approximation given in \eqref{eq:Gaussian_approx_athreya}. Proposition~\ref{prop:Sigma_x} also  shows that if either $d$ or $n-d$ increase with $n$,  then the covariance does not depend on $X_i$, and is given by $\underline{R}^2$. Interestingly, this result establishes a connection between the analysis of spectral methods and the information-theoretic analysis of dense networks given in~\cite{reeves2019}, which also involves a GMM with common covariance across the mixtures.

\subsection{Proposed Method}
\label{sec:method}

Let $\mathbb{P}_0(x,z)$ be the distribution of the pair $(X_0, Z_0)$ where $X_0 \sim P$ and $Z_0$ is conditionally Gaussian:
\begin{align}
 Z_0  \sim \normal\left( (\bar{R}^2 - I) X_0, \Sigma(X_0) \right), \label{eq:Z0}
\end{align}
and $\Sigma(x)$ is defined by \eqref{eq:Sigma_x}. 
Our method has three components:
\begin{enumerate}
\item (Spectral embedding) Let $V$ be the eigenvectors of the rank-$s$ projection of the normalized adjaceny matrix $A - (d/n) \one \one^T$ 
and let $Y_1, \dots, Y_n$ be given by \eqref{eq:scaled_vecs}.

\item (Alignment  via maximum likelihood) Let $U^*$ be a solution to the optimization problem 
\begin{align}
  \max_{U} \prod_{i=1}^n \mathbb{P}_0(U Y_i), \label{eq:max}
\end{align}
where the maximum is over all orthogonal matrices $U$ and $\mathbb{P}_0(z)$ is the marginal of $\mathbb{P}_0(x,z)$ with respect to $x$. 
For an SBM, this is the marginal over a GMM.

\item (Classification)
For $i=1, \dots, n$ output the the posterior $\mathbb{P}_0(x \mid z)$ evaluated on the rotated data $z =(U^* Y_i)$. For an SBM, the posterior is represented by the probability vector $\hat{p}_i = (\hat{p}_{i1}, \dots, \hat{p}_{iK})$.
\end{enumerate}

Besides the eigenvalue decomposition, the potentially computationally challenging step in our method is the optimization with respect to an orthogonal matrix. For convenience, this optimization can be carried out over a restricted set of orthogonal matrices belonging to the set $\{ U \, : \, U \diag(r_1, \dots, r_s) U^T = R\}$. In the simulations that follow, we obtain an approximate solution by searching over a set of representative orthogonal matrices.

\section{Experimental Results}\label{sec:sims}
In this section, we study the behavior of the degree-balanced SBM, parameterized by $(n,p,d,R)$ from Section~\ref{sec:dbsbm}.

\subsection{Numerical Simulations}
We generate a network of $n = 5000$ nodes and $K=3$ communities with probability vector $p = (0.1,0.3,0.6)$. The support of $P$ is defined according to  \cite[Remark 1]{reeves2019}, which yields
\begin{align*}
\mu_1 = \begin{pmatrix} 3 \\ 0 \end{pmatrix}, \quad \mu_2 =  \begin{pmatrix} -1/3 \\2 \sqrt{5}/3 \end{pmatrix}, \quad \mu_2 =   \begin{pmatrix} - 1/3 \\  - \sqrt{5}/3\end{pmatrix}.
\end{align*}
The adjacency matrix is generated according to \eqref{eq:A} with average degree $d = 15$ and 
\begin{align*}
R = U \diag(r_1, r_2) U^T, \qquad U = \frac{1}{\sqrt{2}} \begin{pmatrix}1 & 1\\ 1 & - 1 \end{pmatrix}.
\end{align*}

\begin{figure}[h]
\centering\includegraphics[width=.6\columnwidth]{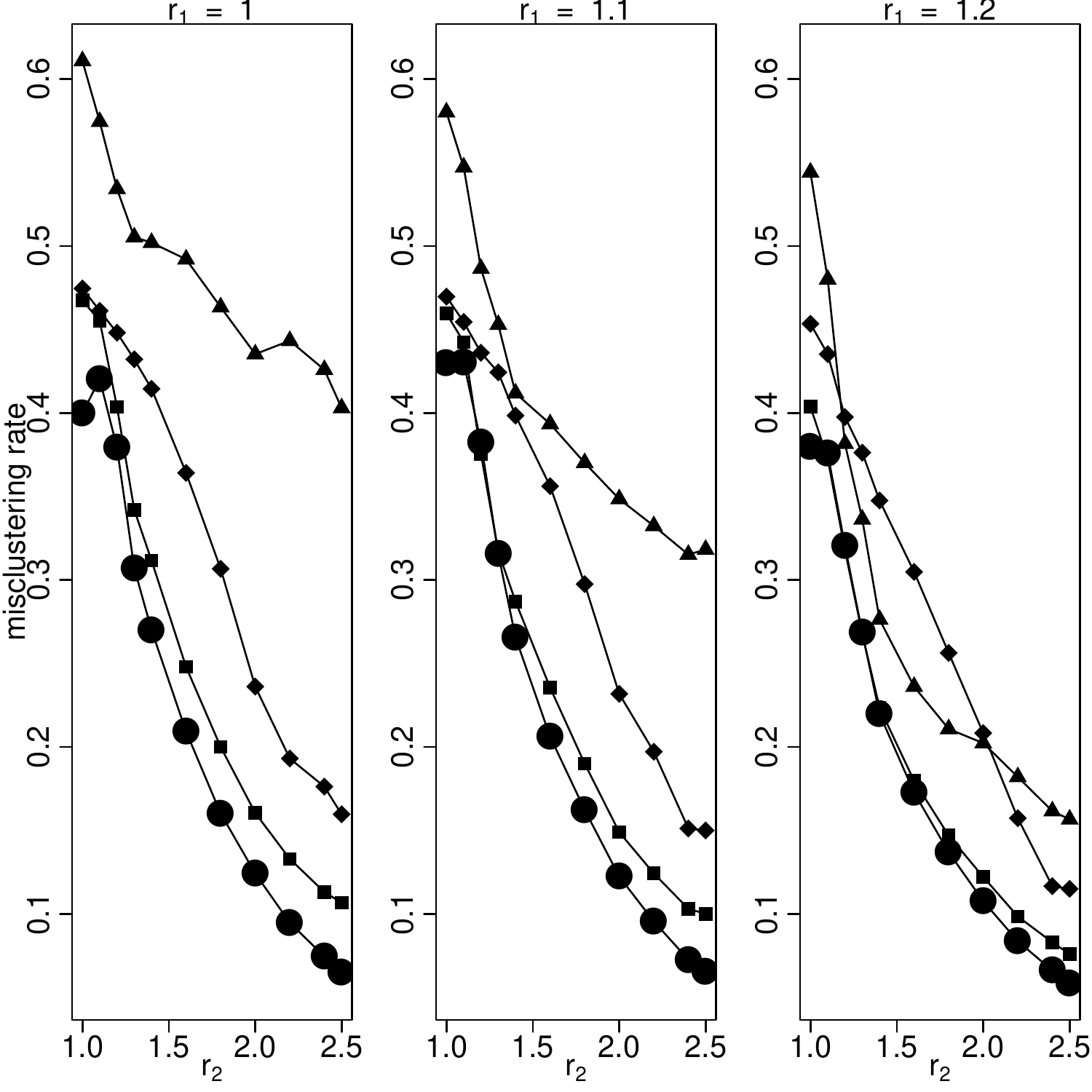}
\caption{Misclustering rate across varying regimes. Each point is the average misclustering rate over $100$ independent networks. Methods are given by:  GMM: $\CIRCLE$, Low-Noise GMM: $\blacksquare$, Uninformed GMM: $\blacklozenge$, $K$-Means: $\blacktriangle$.
\label{fig:justAMisclust}}
\end{figure}

 We compare four different methods: 
\begin{itemize}
\item (GMM) The method described in Section~\ref{sec:method}.

\item (Low-Noise GMM) This is the version of the method described in Section~\ref{sec:method} where $\mathbb{P}_0(x,z)$ corresponds to the low-noise approximation in \eqref{eq:Gaussian_approx_athreya}.

\item (Uninformed GMM) This method fits a GMM to the
the rows of the eigenvectors associated with the $K-1$ largest eigenvalues (in magnitude) of  $A - (d/n) \one \one^T$,
\item ($K$-Means) This method is applied to the same selected eigenvectors as the uniformed GMM.

\end{itemize}
We study the regimes where $(\tilde r_1,\tilde r_2)\in\{1,1.1,1.2\}\times\{1,\dots,2.6\}$ and $(r_1,r_2) = \textrm{sort}(\tilde r_1,\tilde r_2)$.
Performance is assessed using the misclustering rate. For the GMM methods, we use the maximum a posteriori estimate of the community memberships. 
Following the usual convention in the literature, this metric is optimized over permutations of the estimated labels, to mitigate the effects of label switching. 

Figure \ref{fig:justAMisclust} shows the misclustering rate across the different regimes. Each point in the Figure is the average misclustering rate over $100$ networks generated for each set of parameters. As $r_1$ and $r_2$ increase, we see an improvement across all methods. When the eigenvalues of $R$ are close to one, there is almost no correlation between the eigenvectors of the adjacency matrix and the community structure and so all methods perform poorly. Two interesting phenomena can be observed in the figure: first, the relative advantage of our proposed approach to other approaches appears to increase as the larger eigenvalue of $R$ grows (when $r_2=2.5$ we have a nearly 50\% improvement over the next competing method). Second, as the eigenvalues of $R$ grow, the performance of $K$-Means improves and surpasses that of the uninformed GMM.  

\subsection{Real-World Data Analysis}
In this section, we apply our method to an email network from a large European research institution \cite{yin:2017,leskovec2007}. In these data, an undirected edge exists between person $i$ and person $j$ if either one or both had sent an email to the other. Only communication between individuals within the institution is considered and each person belongs to one of $42$ known departments which can be treated as ground truth communities. The smallest 40 communities are combined, yielding a total of $3$ ground truth communities.
This leads to an approximate degree balanced network with the average in each community being around 30. 
We provide two types of analysis for this data: one based on oracle model parameters and one based on estimated ones. Consider that an oracle provides us with the true values of $R$ and $p$ (which we can compute based on our ground truth community information). The misclustering rate of our approach with oracle $R$ and $p$ is $0.2$.

In practice, since one rarely has access to true $R$ and $p$ values, we also estimate $\hat R$ and $\hat p$ based on a $10\%$ sample from the true communities. Using the estimated values of $p$ and $R$ we apply the method of Section~\ref{sec:methodmain} and achieve a misclustering rate of $0.3$ where we have accounted for using $10\%$ of the data to learn the model parameters. These results compare favorably to the performance of $K$-means which achieves a misclustering rate of $0.368$.

\section{Conclusion}
 In this paper, we propose a Gaussian mixture model representation of a projection of the adjacency matrix that can be leveraged for community detection in the {\it non-vanishing noise regime}. In contrast to prior work, we use a model for the joint distribution between the eigenvectors and the latent community structure that includes truncation and shrinkage effects. We demonstrate empirically that this novel representation is able to improve on the performance of community detection in moderate to high noise regimes.
 
 For future directions, it would be interesting to see if these empirical results can be proven rigorously and extended to regularized spectral methods based on the graph Laplacian \cite{zhang2018understanding} or other data-driven techniques \cite{zhang:2016}. This method can further be extended to the settings of multiple observed networks on the same set of units \cite{mayya2019multiple}.

\bibliographystyle{IEEEtran}

\bibliography{ICML_draft}

\begin{thebibliography}{10}
\providecommand{\url}[1]{#1}
\csname url@samestyle\endcsname
\providecommand{\newblock}{\relax}
\providecommand{\bibinfo}[2]{#2}
\providecommand{\BIBentrySTDinterwordspacing}{\spaceskip=0pt\relax}
\providecommand{\BIBentryALTinterwordstretchfactor}{4}
\providecommand{\BIBentryALTinterwordspacing}{\spaceskip=\fontdimen2\font plus
\BIBentryALTinterwordstretchfactor\fontdimen3\font minus
  \fontdimen4\font\relax}
\providecommand{\BIBforeignlanguage}[2]{{%
\expandafter\ifx\csname l@#1\endcsname\relax
\typeout{** WARNING: IEEEtran.bst: No hyphenation pattern has been}%
\typeout{** loaded for the language `#1'. Using the pattern for}%
\typeout{** the default language instead.}%
\else
\language=\csname l@#1\endcsname
\fi
#2}}
\providecommand{\BIBdecl}{\relax}
\BIBdecl

\bibitem{farine2015constructing}
D.~R. Farine and H.~Whitehead, ``Constructing, conducting and interpreting
  animal social network analysis,'' \emph{Journal of Animal Ecology}, vol.~84,
  no.~5, pp. 1144--1163, 2015.

\bibitem{dunbar2015structure}
R.~I. Dunbar, V.~Arnaboldi, M.~Conti, and A.~Passarella, ``The structure of
  online social networks mirrors those in the offline world,'' \emph{Social
  networks}, vol.~43, pp. 39--47, 2015.

\bibitem{stadtfeld2019integration}
C.~Stadtfeld, A.~V{\"o}r{\"o}s, T.~Elmer, Z.~Boda, and I.~J. Raabe,
  ``Integration in emerging social networks explains academic failure and
  success,'' \emph{Proceedings of the National Academy of Sciences}, vol. 116,
  no.~3, pp. 792--797, 2019.

\bibitem{grinberg2019fake}
N.~Grinberg, K.~Joseph, L.~Friedland, B.~Swire-Thompson, and D.~Lazer, ``Fake
  news on twitter during the 2016 us presidential election,'' \emph{Science},
  vol. 363, no. 6425, pp. 374--378, 2019.

\bibitem{Rohe11}
K.~Rohe, S.~Chatterjee, and B.~Yu, ``Spectral clustering and the high
  -dimensional stochastic blockmodel,'' \emph{The Annals of Statistics},
  vol.~39, no.~4, pp. 1878--1915, 2011.

\bibitem{athreya2016limit}
A.~Athreya, C.~E. Priebe, M.~Tang, V.~Lyzinski, D.~J. Marchette, and D.~L.
  Sussman, ``A limit theorem for scaled eigenvectors of random dot product
  graphs,'' \emph{Sankhya A}, vol.~78, no.~1, pp. 1--18, 2016.

\bibitem{Suwan16}
S.~Suwan, D.~Lee, R.~Tang, and et~al, ``Empirical bayes estimation for the
  stochastic blockmodel,'' \emph{Electronic Journal of Statistics}, vol.~10,
  pp. 761--782, 2016.

\bibitem{decelle2011asymptotic}
A.~Decelle, F.~Krzakala, C.~Moore, and L.~Zdeborov{\'a}, ``Asymptotic analysis
  of the stochastic block model for modular networks and its algorithmic
  applications,'' \emph{Physical Review E}, vol.~84, no.~6, p. 066106, 2011.

\bibitem{abbe:2018}
E.~Abbe, ``Community detection and stochastic block models: {R}ecent
  developments,'' \emph{Journal of Machine Learning Research}, vol.~18, no.
  177, pp. 1--86, 2018.

\bibitem{lelarge:2018}
M.~Lelarge and L.~Miolane, ``Fundamental limits of symmetric low-rank matrix
  estimation,'' \emph{Probability Theory and Related Fields}, 2018.

\bibitem{reeves2019}
G.~Reeves, V.~Mayya, and A.~Volfovsky, ``The geometry of community detection
  via the mmse matrix,'' \emph{2019 IEEE International Symposium on Information
  Theory (ISIT)}, 2019.

\bibitem{Holland83}
P.~Holland, K.~Laskey, and S.~Leinhardt, ``Stochastic blockmodels : First
  steps,'' \emph{Social Networks}, vol.~5, pp. 109--137, 1983.

\bibitem{Georges2012}
F.~Benaych-Georges and R.~Nadakuditi, ``The singular values and vectors of low
  rank perturbations of large rectangular random matrices,'' \emph{Journal of
  Multivariate Analysis}, vol. 111, pp. 120--135, 2012.

\bibitem{yin:2017}
H.~Yin, A.~Benson, J.~Leskovec, and D.~Gleich, ``"local higher-order graph
  clustering,'' in \emph{23rd ACM SIGKDD International Conference on Knowledge
  Discovery and Data Mining}, 2017.

\bibitem{leskovec2007}
J.~Leskovec, J.~Kleinberg, and C.~Faloutsos, ``Graph evolution: Densification
  and shrinking diameters,'' \emph{ACM Transactions on Knowledge Discovery from
  Data}, vol.~1, no.~1, 2007.

\bibitem{zhang2018understanding}
Y.~Zhang and K.~Rohe, ``Understanding regularized spectral clustering via graph
  conductance,'' \emph{arXiv preprint arXiv:1806.01468}, 2018.

\bibitem{zhang:2016}
P.~Zhang, ``Robust spectral detection of global structures in the data by
  learning a regularization,'' in \emph{Advances in Neural Information
  Processing Systems}, 2016, pp. 541--549.

\bibitem{mayya2019multiple}
V.~Mayya and G.~Reeves, ``Mutual information in community detection with
  covariate information and correlated networks,'' in \emph{57th Annual
  Allerton Conference on Communication, Control, and Computing}, 2019.

\end{thebibliography}

\end{document}